\documentclass[11pt,reqno]{amsart}

\usepackage[T1]{fontenc}
\usepackage{amsmath,amssymb,mathtools}
\usepackage{microtype}
\usepackage[margin=1.08in]{geometry}
\usepackage{enumitem}
\usepackage{xcolor}
\usepackage{hyperref}
\usepackage[capitalise,noabbrev]{cleveref}

\hypersetup{
  colorlinks=true,
  linkcolor=blue!45!black,
  citecolor=green!35!black,
  urlcolor=blue!55!black,
  pdftitle={Critical Points of Line Restrictions of Signed Point-Charge Potentials}
}

\newtheorem{theorem}{Theorem}[section]
\newtheorem{proposition}[theorem]{Proposition}
\newtheorem{lemma}[theorem]{Lemma}
\newtheorem{corollary}[theorem]{Corollary}

\newtheorem{remark}[theorem]{Remark}
\newtheorem{example}[theorem]{Example}

\newcommand{\RP}{\mathbb{RP}}
\newcommand{\R}{\mathbb{R}}
\newcommand{\Sph}{\mathbb{S}}
\newcommand{\cU}{\mathcal{U}}
\newcommand{\cLS}{\mathcal{LS}}
\newcommand{\ind}{\operatorname{ind}}
\newcommand{\wind}{\operatorname{wind}}
\newcommand{\Eq}{\operatorname{Eq}}
\newcommand{\dist}{\operatorname{dist}}
\newcommand{\sgn}{\operatorname{sgn}}

\title{Critical Points of Line Restrictions of Signed Point-Charge Potentials}
\author{Xiuqing Duan}
\address{School of Physical and Mathematical Sciences, Nanyang Technological University, Singapore}
\email{xiuqing.duan@ntu.edu.sg}
\date{\today}


\begin{document}

\begin{abstract}
Let \(\alpha>0\) and restrict a finite inverse-power potential with
arbitrary real coefficients to a line.  Combine terms having the same
projected centre and squared height \((a,b^2)\), delete classes whose
coefficient sum is zero, and let \(m\) be the number of remaining classes.
We prove the following dichotomy.  If \(m=0\), exact cancellation occurs if
and only if every class sum vanishes, and every ordinary point of the
original domain is critical.  If \(m\geq1\), there are at most \(2m-1\)
critical points: when all effective heights are positive the zeros are
counted with analytic multiplicity, while in the presence of effective
sources on the line the assertion is one global distinct-point bound.  
This proves the signed line conjecture of
Gabrielov--Novikov--Shapiro, valid throughout their range and in fact for
every \(\alpha>0\).  The structural input is a projective paired Haar
theorem: for finite \(\beta>1\), the full \(2m\)-dimensional space
\(\sum L_j/Q_j^\beta\), with pairwise nonproportional positive-definite
binary quadratics and arbitrary real linear numerators, has at most
\(2m-1\) projective zeros counted with multiplicity.  For positive charges,
the substitution \(p=2\alpha\) proves Conjecture~3 of
Edelsbrunner--Fillmore--Oliveira throughout its stated range \(p\geq1\) and extends the same conclusion to every \(p>0\).
For every \(n\), an
explicit positive \(n\)-charge configuration attains \(2n-1\) simple critical
points.
\end{abstract}

\maketitle

\section{Introduction}

On a prescribed line, a signed inverse-power
restriction has the form
\begin{equation}\label{eq:intro-signed}
  F_\alpha(t)=\sum_{i=1}^{\ell}
  \eta_i\bigl((t-a_i)^2+b_i^2\bigr)^{-\alpha},
  \qquad \alpha>0,
\end{equation}
where \(\eta_i\in\R\), \(a_i\in\R\), \(b_i\geq0\), and \(\ell\geq1\).
A point \(t \in \R\) is ordinary if \(t \neq a_i\) for every \(i\) with \(b_i=0\).
The question is to bound the zeros of \(F_\alpha'\) on the ordinary domain.
This is a tangential critical-point problem, not the unrestricted
electrostatic equilibrium problem: a critical point of a line restriction
need not be a critical point of the ambient potential.

The first issue is algebraic rather than analytic.  
Terms arising from distinct charge locations can induce the same restricted kernel.
Such terms must be combined, and a
class whose coefficient sum vanishes must be removed.  We make this
reduction canonical below and distinguish the original domain from the possibly
larger reduced domain obtained by filling canceled punctures.
The resulting number \(m\) of nonzero effective kernels, rather than the 
number \(\ell\) of terms in the original representation, is the natural parameter in the theorem.

Gabrielov, Novikov, and Shapiro formulated their line-restriction conjecture
for arbitrary real coefficients and \(\alpha\geq\tfrac12\); it is
Conjecture~1.9 in \cite{GabrielovNovikovShapiro2007}.  
Its literal wording permits 
restrictions that vanish identically by exact cancellation, so any finite bound 
must treat that case separately from the \(m\geq1\) case.
Edelsbrunner, Fillmore, and Oliveira (EFO) later stated the
positive-charge specialization for \(p\geq1\) as their Conjecture~3
\cite{EdelsbrunnerFillmoreOliveira2026} (where they attribute this conjecture to Gabrielov, Novikov, and Shapiro).
Their coefficients are positive,
and their exponent is related to ours by \(p=2\alpha\); the EFO conjecture
is therefore a positive specialization, not an equivalent signed
formulation.  We prove the corrected signed statement for every
\(\alpha>0\), and hence the positive statement for every \(p>0\).

The structural theorem is a projective Haar bound for a paired family.  
A line derivative determines one linear numerator over each quadratic denominator; it belongs to 
the larger paired space obtained by allowing an arbitrary real linear numerator over each denominator.
We prove that the full \(2m\)-dimensional space has
the \(2m-1\) zero bound on \(\RP^1\), with multiplicity.  This is a global
Haar/extended-Chebyshev assertion for the full space;
it does not assert that every ordering of the generators 
forms an extended complete Chebyshev system (ECT-system).

Our terminology follows the classical theory of Tchebycheff and Haar spaces
\cite{KarlinStudden1966}.  Several adjacent results do not directly imply
the paired theorem.  Strict positive definiteness of inverse-multiquadric
interpolation matrices \cite{Micchelli1986} is weaker than total positivity
and does not control zeros at arbitrary test nodes; this distinction is
made explicit in \cite{KarvonenKanagawaSarkka2019}, with classical
total-positivity background in \cite{Karlin1968}.  Prescribed-pole
Cauchy--Vandermonde ECT systems cover the rational endpoint \(\beta=1\), but not the
paired family for \(\beta>1\) needed here \cite{Muhlbach2000}; known total-positivity
results for a Cauchy kernel concern a different one-parameter kernel
\cite{Simon2014}.  Positive mixtures of two Cauchy densities form a
previously classified special case \cite{Dosla2009}.  To our knowledge, the
global projective Haar bound for the full variable-quadratic paired space has
not appeared before.  The proof is a new application of convex-curve methods
of M.~Z. Shapiro and Anisov after a family-specific Wronskian and
rational-normal endpoint analysis.  For comparison, the compatible fixed-Spin
convention recorded by Saldanha--B. Shapiro and Alves is also documented;
this is not a claim of a new general theory of Chebyshev systems.

The proof has two structural steps.  Gegenbauer generating functions first
identify the paired Wronskian with a gradient-evaluation determinant.  A
vanishing determinant would produce too many critical points of a
generalized-axisymmetric harmonic polynomial; a planar index theorem rules
this out.  An exponent homotopy then connects the locally convex paired
curve to the rational normal curve.  
Anisov's closed-projective
loss-of-convexity theorem supplies the global step; 
the fixed-Spin
component formulation recorded by Saldanha--B. Shapiro and Alves gives a
compatible but non-load-bearing interval description
\cite{Shapiro1994,Anisov1998,SaldanhaShapiro2012,Alves2016}.  Smooth signed
restrictions follow directly.  Effective sources on the line are treated by
one globally coordinated two-sided smoothing, not by adding componentwise
estimates.

In independent work submitted shortly after the first version of the present paper, 
Oliveira \cite{Oliveira2026} proves the same \(2m-1\) bound for distinct critical points of signed line restrictions for every \(p>0\). 
Both approaches use paired families and planar index arguments, but Oliveira uses duality with an auxiliary potential, whereas we use Gegenbauer jets and projective convexity.
Our smoothing argument also yields the refined singular derivative bound \(2m-r-1\) in Remark \ref{rem:count-singular}, where \(r\) is the number of effective poles.

\section{Reduction and main results}
\label{sec:main-results}

Fix \(\alpha>0\), \(\ell\geq1\), and the data in
\eqref{eq:intro-signed}.  The original regular domain is
\begin{equation}\label{eq:original-domain}
  D_{\rm orig}=\R\setminus\{a_i:b_i=0\}.
\end{equation}
Define
\begin{equation}\label{eq:kernel-equivalence}
  i\sim j
  \quad\Longleftrightarrow\quad
  (a_i,b_i^2)=(a_j,b_j^2).
\end{equation}
Because the heights \(b_i\) and \(b_j\) are chosen nonnegative, this is equivalent to equality
of \((a_i,b_i)\), but \((a_i,b_i^2)\) is the invariant visible in the
restricted kernel.  
Let \(\mathfrak C\) be the set of equivalence classes.
For each \(C\in\mathfrak C\), let \((a_C,b_C^2)\) denote the common value of \((a_i,b_i^2)\) for \(i\in C\), and define
\begin{equation}\label{eq:class-sum}
  \widehat\eta_C=\sum_{i\in C}\eta_i.
\end{equation}
Let \(\mathfrak C_*=\{C\in\mathfrak C:\widehat\eta_C\neq0\}\), and put
\(m=|\mathfrak C_*|\).  If \(m\geq1\), reindex the active classes by
\(j=1,\ldots,m\), and denote their nonzero coefficients and distinct
parameters by \(c_j\) and \((a_j,b_j^2)\).  If \(m=0\), the following sum
is understood to be empty.  Define the reduced extension
\begin{equation}\label{eq:reduced-restriction}
  F_{\rm red}(t)=\sum_{j=1}^{m}
  c_j\bigl((t-a_j)^2+b_j^2\bigr)^{-\alpha}
\end{equation}
on
\begin{equation}\label{eq:reduced-domain}
  D_{\rm red}=\R\setminus\{a_j:b_j=0,\ 1\leq j\leq m\}.
\end{equation}
Then
\begin{equation}\label{eq:domain-inclusion}
  D_{\rm orig}\subseteq D_{\rm red},
  \qquad F_\alpha=F_{\rm red}\quad\hbox{on }D_{\rm orig}.
\end{equation}
Thus reduction fills only punctures produced by classes that cancel
exactly.

\begin{theorem}[Reduced signed Maxwell-line dichotomy]
\label{thm:signed-line}
With the preceding notation,
\begin{equation}\label{eq:exact-cancellation-equivalences}
  m=0
  \quad\Longleftrightarrow\quad
  \widehat\eta_C=0\ \hbox{for every }C\in\mathfrak C
  \quad\Longleftrightarrow\quad
  F_\alpha\equiv0\ \hbox{on }D_{\rm orig}.
\end{equation}
Moreover, exactly one of the following alternatives holds.
\begin{enumerate}[label=\textup{(\roman*)}]
\item If \(m=0\), then \(F_{\rm red}\) is the zero function, and every point
of \(D_{\rm orig}\) is critical.
\item If \(m\geq1\) and every effective height is positive, then
\(D_{\rm red}=\R\) and
\begin{equation}\label{eq:smooth-signed-bound}
  \sum_{t\in\R:\,F_{\rm red}'(t)=0}
  \operatorname{ord}_t(F_{\rm red}')\leq2m-1.
\end{equation}
\item If \(m\geq1\) and at least one effective height is zero, then the
single global distinct-point estimate
\begin{equation}\label{eq:singular-signed-bound}
  \sum_{I\in\pi_0(D_{\rm red})}
  \#\{t\in I:F_{\rm red}'(t)=0\}\leq2m-1
\end{equation}
holds. Here \(\pi_0(D_{\rm red})\) denotes the set of connected components of \(D_{\rm red}\). 
No multiplicity assertion is made in this singular case.
\end{enumerate}
In the two nonzero alternatives, \(m\leq\ell\), and hence
\begin{equation}\label{eq:presentation-bound}
  2m-1\leq2\ell-1.
\end{equation}
The critical points on \(D_{\rm orig}\) are a subset of those counted on
\(D_{\rm red}\), so the applicable bound also holds on the original
domain.
\end{theorem}

\begin{remark}[Scope of the dichotomy]\label{rem:signed-scope}
No sign, total-charge, moment, genericity, or off-line assumption occurs in
\cref{thm:signed-line}.  It includes mixed signs, zero total charge,
vanishing lower moments, degenerate smooth critical points, sources on the
restricting line, and every finite \(\alpha>0\).  
An effective pole--that is, a point \(t=a_j\) with \(b_j=0\) in the reduced data--lies outside both \(D_{\rm orig}\) and \(D_{\rm red}\) and is never counted as critical.
\end{remark}

A real binary quadratic form \(Q\) is \emph{positive definite} when
\(Q(u,v)>0\) for every \((u,v)\neq(0,0)\).  For pairwise
nonproportional such forms \(Q_1,\ldots,Q_n\) and \(\beta>0\), put
\begin{equation}\label{eq:paired-space}
  \cU_\beta(Q_1,\ldots,Q_n)
  =\left\{
    \sum_{i=1}^n\frac{L_i(u,v)}{Q_i(u,v)^\beta}:
    L_i\in\R[u,v]_1
  \right\}.
\end{equation}
Here \(\mathbb R[u,v]_1\) denotes the two-dimensional space of homogeneous real linear forms \(Au+Bv\).
On the unit circle, these expressions are smooth and anti-periodic, and
therefore define sections of the tautological real line bundle on
\(\RP^1\).  Their zero sets and local orders descend to \(\RP^1\); a zero
at infinity is measured in the complementary affine chart.

\begin{theorem}[Paired Haar theorem]\label{thm:paired-haar}
Let \(\beta>1\) be finite, and let \(Q_1,\ldots,Q_n\) be pairwise
nonproportional positive-definite real binary quadratic forms.  Then
\(\cU_\beta(Q_1,\ldots,Q_n)\) has dimension \(2n\), and every nonzero
member has at most \(2n-1\) zeros on \(\RP^1\), counted with multiplicity.
\end{theorem}

This is a global projective Haar/extended-Chebyshev theorem for the full
paired space.  Pairwise nonproportionality is the only distinctness
hypothesis, arbitrary real linear numerators are allowed, and no ordering of
the quadratics is used.  Rank-one quadratics are excluded and will be
handled separately in the singular part of \cref{thm:signed-line}.

We prove \cref{thm:paired-haar} in
\cref{sec:local-wronskian,sec:index,sec:globalization}; the signed theorem is
then proved in \cref{sec:signed-proof}.  Its physical consequences and
examples appear in \cref{sec:physical-consequences,sec:examples}.

\section{Gegenbauer jets and the paired Wronskian}
\label{sec:local-wronskian}

Fix an arbitrary point of \(\RP^1\), choose a local affine coordinate
\(s\) that vanishes there, and use \(1,s\) as a numerator basis.  Dividing
each quadratic by its positive value at \(s=0\) gives
\begin{equation}\label{eq:q-normal}
  q_i(s)=1+2X_i s+Y_i s^2,
  \qquad Y_i>X_i^2.
\end{equation}
The normalized pairs \((X_i,Y_i)\) are distinct: equality of two normalized
quadratics would make the original forms proportional.  The local paired
family is
\begin{equation}\label{eq:local-pairs}
  q_i(s)^{-\beta},\qquad s q_i(s)^{-\beta},
  \qquad i=1,\ldots,n.
\end{equation}

Put \(\lambda=\beta-1>0\).  The Gegenbauer generating function is
\begin{equation}\label{eq:gegenbauer-generating}
  (1-2zt+t^2)^{-\lambda}
  =\sum_{r\geq0}C_r^\lambda(z)t^r;
\end{equation}
see DLMF~18.12.4 \cite{NIST:DLMF}, or \cite[Sec.~4.7]{Szego1975}.
Here \(C_r^\lambda\) denotes the Gegenbauer polynomial of degree \(r\) with parameter \(\lambda\), in the normalization determined by this generating function.
With \(z=X/\sqrt Y\) and \(t=-\sqrt Y\,s\), it gives
\begin{equation}\label{eq:K-generating}
  (1+2Xs+Ys^2)^{-\lambda}
  =\sum_{r\geq0}(-1)^rK_r^\lambda(X,Y)s^r,
\end{equation}
where
\begin{equation}\label{eq:K-definition}
  K_r^\lambda(X,Y)
  =Y^{r/2}C_r^\lambda\!\left(\frac{X}{\sqrt Y}\right).
\end{equation}
Parity of the Gegenbauer polynomial makes \(K_r^\lambda\) a polynomial.
More explicitly,
\begin{equation}\label{eq:K-explicit}
  K_r^\lambda(X,Y)
  =\sum_{j=0}^{\lfloor r/2\rfloor}
    \frac{(-1)^j2^{r-2j}(\lambda)_{r-j}}
         {j!(r-2j)!}
    X^{r-2j}Y^j.
\end{equation}
Here \((\lambda)_0=1\) and
\[
(\lambda)_k=\lambda(\lambda+1)\cdots(\lambda+k-1)
\qquad(k\geq1)
\]
is the rising Pochhammer symbol.

Writing
\[
q=q(X,Y;s)=1+2Xs+Ys^2,
\]
differentiation of the left side of \eqref{eq:K-generating} with respect to \(X\) and \(Y\) gives
\begin{equation}\label{eq:q-parameter-derivatives}
  \partial_Xq^{-\lambda}=-2\lambda s q^{-\beta},
  \qquad
  \partial_Yq^{-\lambda}=-\lambda s^2q^{-\beta}.
\end{equation}
Let \(W_\beta(0)\) be the \(2n\)-by-\(2n\) Wronskian whose rows are
derivative orders \(0,\ldots,2n-1\) and whose columns, in paired order,
are the functions in \eqref{eq:local-pairs}.  Define
\(A\in\R^{2n\times2n}\), with rows ordered
\((1,X),(1,Y),\ldots,(n,X),(n,Y)\) and columns
\(r=1,\ldots,2n\), by
\begin{equation}\label{eq:A-definition}
  A_{(i,X),r}=\partial_XK_r^\lambda(X_i,Y_i),
  \qquad
  A_{(i,Y),r}=\partial_YK_r^\lambda(X_i,Y_i).
\end{equation}

\begin{lemma}[Exact jet determinant]\label{lem:exact-factor}
With the preceding ordering,
\begin{equation}\label{eq:exact-factor}
  \boxed{
  W_\beta(0)=
  \frac{(-1)^n\prod_{r=0}^{2n-1}r!}
       {2^n\lambda^{2n}}\det A.}
\end{equation}
\end{lemma}

\begin{proof}
Write \(q_i(s)^{-\beta}=\sum_{r\geq0}c_{i,r}s^r\), with
\(c_{i,-1}=0\).  Coefficient comparison in
\eqref{eq:q-parameter-derivatives} yields
\[
  c_{i,r}=\frac{(-1)^r}{2\lambda}
    \partial_XK_{r+1}^\lambda(X_i,Y_i),
  \qquad
  c_{i,r-1}=\frac{(-1)^r}{\lambda}
    \partial_YK_{r+1}^\lambda(X_i,Y_i).
\]
Passing from Taylor coefficients to derivatives multiplies row \(r\) by
\(r!\).  The alternating signs contribute
\[
  (-1)^{\sum_{r=0}^{2n-1}r}
  =(-1)^{n(2n-1)}=(-1)^n.
\]
The first column in each pair contributes \((2\lambda)^{-1}\), and the
second contributes \(\lambda^{-1}\).  Transposition does not change the
determinant, and \eqref{eq:exact-factor} follows.
\end{proof}
The factor in \eqref{eq:exact-factor} is nonzero for every \(\lambda>0\).
We next show that \(\det A\) cannot vanish.

\begin{proposition}[A singular Wronskian creates critical points]
\label[proposition]{prop:singular-creates-critical}
If \(W_\beta(0)=0\), there is a nonconstant polynomial \(h(x,\rho)\),
even in \(\rho\), of exact degree \(D\leq2n\), satisfying
\begin{equation}\label{eq:generalized-axisymmetric-pde}
  L_\lambda h
  :=h_{xx}+h_{\rho\rho}+\frac{2\lambda}{\rho}h_\rho=0,
\end{equation}
and having the \(n\) distinct upper-half-plane critical points
\begin{equation}\label{eq:critical-data-points}
  \left(X_i,\sqrt{Y_i-X_i^2}\right),
  \qquad i=1,\ldots,n.
\end{equation}
\end{proposition}

\begin{proof}
If \(W_\beta(0)=0\), then \(A\) is singular.  Hence some nonzero vector
\((d_1,\ldots,d_{2n})\) makes
\begin{equation}\label{eq:H-definition}
  H(X,Y)=\sum_{r=1}^{2n}d_rK_r^\lambda(X,Y)
\end{equation}
satisfy
\begin{equation}\label{eq:H-critical}
  H_X(X_i,Y_i)=H_Y(X_i,Y_i)=0
  \qquad(i=1,\ldots,n).
\end{equation}
Set
\begin{equation}\label{eq:h-substitution}
  h(x,\rho)=H(x,x^2+\rho^2).
\end{equation}
Then
\[
  h=\sum_{r=1}^{2n}d_r Z_r^\lambda,
  \qquad
  Z_r^\lambda(x,\rho)
  =(x^2+\rho^2)^{r/2}
  C_r^\lambda\!\left(\frac{x}{\sqrt{x^2+\rho^2}}\right).
\]
Each \(Z_r^\lambda\) is a nonzero homogeneous polynomial of degree \(r\):
indeed \(C_r^\lambda(1)=(2\lambda)_r/r!>0\).  Thus \(h\) is nonconstant
and has degree at most \(2n\).

The Gegenbauer differential equation
\[
  (1-z^2)(C_r^\lambda)''-(2\lambda+1)z(C_r^\lambda)'
  +r(r+2\lambda)C_r^\lambda=0
\]
is equivalent, after writing \(x=R\cos\theta\) and
\(\rho=R\sin\theta\), to \(L_\lambda Z_r^\lambda=0\); see
DLMF Table~18.8.1 \cite{NIST:DLMF}.  Hence
\eqref{eq:generalized-axisymmetric-pde} holds.  
Finally, with \(H_X\) and \(H_Y\) evaluated at \((X,Y)=(x,x^2+\rho^2)\), the chain rule gives
\[
  h_x=H_X+2xH_Y,
  \qquad h_\rho=2\rho H_Y.
\]
Equations \eqref{eq:H-critical} and \(Y_i>X_i^2\) show that every point listed in
\eqref{eq:critical-data-points} is a critical point of \(h\).  
These points are distinct because the pairs \((X_i,Y_i)\) are distinct.
\end{proof}

\section{Critical indices for generalized-axisymmetric polynomials}
\label{sec:index}

The next theorem is the local obstruction that forces the paired Wronskian
to be nonzero.  The vector-field index is the Brouwer degree of the
normalized vector field on a small positively oriented circle.

\begin{theorem}[Critical-index theorem]\label{thm:critical-index}
Let \(\lambda>0\), and let \(h(x,\rho)\) be a nonconstant real polynomial,
even in \(\rho\), of exact degree \(D\), satisfying \(L_\lambda h=0\).
Every finite critical point of \(h\), including every degenerate one, is
isolated.  If its first nonconstant homogeneous Taylor term has degree
\(k\geq2\), then
\begin{equation}\label{eq:index-local}
  \ind_q(\nabla h)=1-k\leq-1.
\end{equation}
If \(N_+(h)\) is the number of distinct critical points in \(\rho>0\), then
\begin{equation}\label{eq:critical-count}
  N_+(h)\leq\left\lfloor\frac{D-1}{2}\right\rfloor.
\end{equation}
\end{theorem}

\begin{proof}
Because \(h\) is even in \(\rho\), its derivative \(h_\rho\) is divisible
by \(\rho\); thus \(L_\lambda h=0\) is a polynomial identity across the
axis, not merely an equation on \(\rho\neq0\).  We treat off-axis and axis
points separately, then apply the planar index sum.

\smallskip
\noindent\emph{Off the axis.}
Let \(q=(x_0,\rho_0)\) be a critical point with \(\rho_0\neq0\), and let
\(A_k\), \(k\geq2\), be the first nonconstant homogeneous Taylor term of
\(h-h(q)\) at \(q\).  The coefficient \(2\lambda/\rho\) is analytic near
\(q\).  The degree-\((k-2)\) part of \(L_\lambda h=0\) is therefore
\begin{equation}\label{eq:off-axis-harmonic}
  \Delta A_k=(\partial_x^2+\partial_\rho^2) A_k=0.
\end{equation}
Thus \(A_k\) is a nonzero homogeneous planar harmonic polynomial.  In
polar coordinates it is a nonzero linear combination of
\(r^k\cos(k\theta)\) and \(r^k\sin(k\theta)\).  Its gradient is nonzero on
the punctured plane and winds \(1-k\) times around the origin.  On a
sufficiently small circle about \(q\), the Taylor remainder gives a
homotopy through nonzero vector fields from \(\nabla h\) to \(\nabla A_k\).
The same leading-term domination holds in a punctured disk.  Hence \(q\) is
isolated and \eqref{eq:index-local} holds.

\smallskip
\noindent\emph{On the axis.}
Translation in \(x\) preserves the equation, so take the critical point to
be \((0,0)\).  Evenness gives an expansion
\begin{equation}\label{eq:axis-expansion}
  h(x,\rho)=\sum_{j\geq0}b_j(x)\rho^{2j}.
\end{equation}
Comparing coefficients of \(\rho^{2j}\) in \(L_\lambda h=0\) yields
\begin{equation}\label{eq:axis-recurrence}
  \boxed{
  b_{j+1}(x)=
  -\frac{b_j''(x)}{2(j+1)(2j+1+2\lambda)}.}
\end{equation}
Every denominator is positive.  Consequently, the space of even
homogeneous polynomial solutions of degree \(k\) is one-dimensional: the
coefficient of \(x^k\) determines successively all coefficients of
\(x^{k-2j}\rho^{2j}\).  The nonzero polynomial \(Z_k^\lambda\) defined in
the proof of \cref{prop:singular-creates-critical} spans this space.  The
first nonconstant homogeneous Taylor term at the axis point is therefore a
nonzero multiple of \(Z_k^\lambda\), for some \(k\geq2\).

For \(\lambda>0\), the Gegenbauer polynomial \(C_k^\lambda\) has exactly
\(k\) simple roots in \((-1,1)\), and
\(C_k^\lambda(\pm1)\neq0\).  This follows from Gegenbauer orthogonality
with positive weight \((1-z^2)^{\lambda-1/2}\) and the standard simple-zero
theorem for orthogonal polynomials; see DLMF Table~18.3.1 and
Section~18.2(vi) \cite{NIST:DLMF}.

On the unit circle, the radial and angular components of
\(\nabla Z_k^\lambda\) are
\begin{equation}\label{eq:radial-angular}
  kC_k^\lambda(\cos\theta),
  \qquad
  -\sin\theta\,(C_k^\lambda)'(\cos\theta).
\end{equation}
With the Cartesian vector \((a,b)\) encoded as \(a+ib\), this becomes
\begin{equation}\label{eq:gradient-complex}
  (Z_k^\lambda)_x+i(Z_k^\lambda)_\rho
  =e^{i\theta}q_k(\theta),
\end{equation}
where
\begin{equation}\label{eq:qk-definition}
  q_k(\theta)=
  kC_k^\lambda(\cos\theta)
  -i\sin\theta\,(C_k^\lambda)'(\cos\theta).
\end{equation}
This function never vanishes: off the axis, simultaneous vanishing would
give a multiple interior root; on the axis, the radial component is
nonzero.

We compute its winding directly.  Let
\(\widehat q_k=q_k/|q_k|:S^1\to S^1\).  The regular value \(+i\) has exactly
\(k\) inverse images.  Indeed, each root \(\xi\) of \(C_k^\lambda\)
produces two points with \(\cos\theta=\xi\), and the imaginary parts of
\(q_k\) at those points have opposite signs, so exactly one maps to
\(+i\).  Write \(q_k=a+ib\).  At such a preimage, \(a=0\), \(b>0\), and
\[
  a'(\theta)=-k\sin\theta\,(C_k^\lambda)'(\cos\theta)=kb.
\]
Therefore
\[
  \frac{d}{d\theta}\arg q_k
  =\frac{ab'-ba'}{a^2+b^2}=-k<0.
\]
Each preimage has local degree \(-1\), so
\begin{equation}\label{eq:qk-winding}
  \wind(q_k)=-k,
  \qquad
  \wind(\nabla Z_k^\lambda)=1-k.
\end{equation}
Leading-term domination on a small circle and punctured disk transfers this
winding to \(\nabla h\), proving isolation and \eqref{eq:index-local} at an
axis point as well.

\smallskip
\noindent\emph{The global index sum.}
Let \(h_D\) be the leading homogeneous part of \(h\).  Homogeneous degrees
separate in \(L_\lambda h=0\), and the one-dimensionality just proved shows
that \(h_D\) is a nonzero multiple of \(Z_D^\lambda\).  Its gradient is
nonzero on the unit circle, and \eqref{eq:qk-winding} gives winding
\(1-D\).  Uniformly for \(\omega\in S^1\),
\begin{equation}\label{eq:infinity-leading}
  R^{1-D}\nabla h(R\omega)=\nabla h_D(\omega)+O(R^{-1}).
\end{equation}
Thus \(\nabla h\) is nonzero outside a sufficiently large disk, and its
winding on the boundary of that disk is \(1-D\).  All finite zeros are
isolated and bounded, hence finite in number.  The Brouwer-degree excision
formula, valid also for isolated degenerate zeros, gives
\begin{equation}\label{eq:index-sum}
  1-D=\sum_{\nabla h(q)=0}\ind_q(\nabla h).
\end{equation}
We use the standard local-degree and index-sum formulation from
\cite[Sec.~6]{Milnor1997}.

Evenness in \(\rho\) reflects every upper-half-plane critical point to a
distinct lower-half-plane critical point.  Every finite critical point,
including any axis point, has index at most \(-1\).  Hence
\[
  1-D\leq-2N_+(h),
\]
which is equivalent to \eqref{eq:critical-count}.
\end{proof}

\begin{corollary}[Wronskian nonvanishing]\label[corollary]{cor:wronskian-nonzero}
For every finite \(\beta>1\), the full paired Wronskian is nonzero at every
point of \(\RP^1\).
\end{corollary}

\begin{proof}
If the Wronskian vanished at the arbitrary point normalized in
\eqref{eq:q-normal}, \cref{prop:singular-creates-critical} would give a
polynomial of degree \(D\leq2n\) with at least \(n\) upper-half-plane
critical points.  By \cref{thm:critical-index},
\[
  n\leq N_+(h)
  \leq\left\lfloor\frac{D-1}{2}\right\rfloor
  \leq n-1,
\]
a contradiction.  The local coordinate was chosen at an arbitrary
projective point.  Choosing the other affine chart gives the same
normalization at the point at infinity, so no projective point is omitted.
\end{proof}

In particular, the \(2n\) local paired functions are linearly independent
for every \(\beta>1\).  This proves the dimension assertion in
\cref{thm:paired-haar}; the global zero bound requires one further step.

\section{From local to global convexity}
\label{sec:globalization}

\subsection{The rational-normal endpoint}

At \(\beta=1\), put \(R=\prod_{i=1}^nQ_i\).  It is positive on
\(\RP^1\), and multiplication by \(R\) identifies the paired span with
\begin{equation}\label{eq:beta-one-space}
  \left\{
    \sum_{i=1}^nL_i(u,v)\prod_{j\neq i}Q_j(u,v):
    L_i\in\R[u,v]_1
  \right\}.
\end{equation}

A closed projective curve in \(\mathbb{RP}^d\) is called globally convex if every projective hyperplane intersects it with total multiplicity at most \(d\).

\begin{proposition}\label[proposition]{prop:rational-normal}
The \(2n\) generators in \eqref{eq:beta-one-space} form a basis of the
homogeneous binary forms of degree \(2n-1\).  Consequently the paired curve
at \(\beta=1\) is projectively equivalent to the rational normal curve and
is globally convex.
\end{proposition}

\begin{proof}
A positive-definite real binary quadratic is irreducible over \(\R\), hence
prime in the unique-factorization domain \(\R[u,v]\).  Reduce a linear
relation among the generators modulo \(Q_i\).  Pairwise
nonproportionality makes \(Q_i\) coprime to every \(Q_j\), \(j\neq i\), so
\(Q_i\mid L_i\).  Since \(L_i\) is linear, it follows that \(L_i=0\).
Thus all \(2n\) generators are independent.  The space of homogeneous
binary forms of degree \(2n-1\) also has dimension \(2n\), so the two
spaces coincide.  A nonzero degree-\((2n-1)\) binary form has at most
\(2n-1\) real projective zeros counted with multiplicity, which is the
global convexity of the rational normal curve.
\end{proof}

\subsection{Convexity persistence}

For comparison only, we record the standard fixed-endpoint component
statement.  Neither this theorem nor its
interval-endpoint convention is used in the closed-projective proof below.

\begin{theorem}[Fixed-Spin component theorem]
\label{thm:convex-component}
Let \(d\geq2\) and \(z\in\operatorname{Spin}(d+1)\).  Let
\(\cLS^d(z)\) be the space of positive locally convex curves
\([0,1]\to\Sph^d\) whose initial Frenet frame is \(I\) and whose terminal
lifted Frenet frame is \(z\).  If \(\cLS^d(z)\) contains globally convex
curves, it has exactly two connected components; the globally convex curves
form one of them, and that component is contractible.  If it contains no
globally convex curve, it is connected.
\end{theorem}

Here local convexity means that
\(\det(\gamma,\gamma',\ldots,\gamma^{(d)})\) never vanishes, with positive
local convexity denoting the positive sign.  Global convexity means at most
\(d\) intersections, counted with multiplicity, with every ambient
hyperplane; interval endpoints are not counted.  This exact fixed-Spin
formulation is recorded as Fact~1 on page~15 of
Saldanha--B. Shapiro \cite{SaldanhaShapiro2012} and as Theorem~37 of
Alves \cite{Alves2016}; both present it as a known result.  The interval
proof ingredients are in M.~Z. Shapiro: Lemma~1, Lemmas~6--7,
Theorems~5--6, and Corollary~1 to Theorem~6 \cite{Shapiro1994}.  Lemma~7 is
needed for the boundary-oscillatory endpoint cases under the convention that
interval endpoints are not counted.  Anisov treats closed projective curves,
not arbitrary fixed-Spin intervals \cite{Anisov1998}.

For the anti-periodic family below, we instead use the following consequence of Anisov's theorem.

\begin{corollary}[Closed-projective convexity persistence]
\label[corollary]{cor:closed-projective-persistence}
Let \(d\geq2\), and let
\(s\mapsto\bar\gamma_s:S^1\to\RP^d\), \(0\leq s\leq1\), be a continuous
\(C^{d+1}\)-family of smooth embedded closed projective curves such that, for every local lift \(\gamma_s\),
\[
\det(\gamma_s,\gamma_s',\ldots,\gamma_s^{(d)})
\]
is nowhere zero; equivalently, the curves have no flattenings.
If \(\bar\gamma_0\) is globally convex, then every \(\bar\gamma_s\) is globally
convex.
\end{corollary}

\begin{proof}
As Anisov notes immediately after Theorem~6, global convexity is open
in the \(C^d(S^1,\RP^d)\) topology
\cite[discussion following Theorem~6]{Anisov1998}.
Let
\[
  B=\{s\in[0,1]:\bar\gamma_s
      \text{ is not globally convex}\}.
\]
If \(B\neq\varnothing\), then openness of global convexity makes \(B\)
closed.  Hence \(s_*=\min B\) exists.  Since
\(\bar\gamma_0\) is globally convex, openness also gives \(s_*>0\);
by the definition of \(s_*\), every \(\bar\gamma_s\) with \(s<s_*\)
is globally convex.  Anisov's Theorem~6, applied after translating
the parameter by \(s_*\), gives a flattening of multiplicity at least
two on \(\bar\gamma_{s_*}\), contradicting the nowhere-vanishing jet
determinant.  Thus \(B=\varnothing\).

\end{proof}

\subsection{The exponent homotopy}

Fix an arbitrary finite \(\beta_*>1\).  The case \(n=1\) is immediate,
because a nonzero \(L/Q^\beta\) has the single projective zero of \(L\).
Assume \(n\geq2\), and set
\[
  N=2n,
  \qquad d=N-1=2n-1,
  \qquad x(\theta)=(\cos\theta,\sin\theta).
\]
For \(1\leq\beta\leq\beta_*\), define
\begin{equation}\label{eq:paired-curve}
  \Gamma_\beta(\theta)=
  \left(
    \frac{x(\theta)}{Q_1(x(\theta))^\beta},\ldots,
    \frac{x(\theta)}{Q_n(x(\theta))^\beta}
  \right)\in\R^N\setminus\{0\},
\end{equation}
and let \(\gamma_\beta=\Gamma_\beta/\|\Gamma_\beta\|\in\Sph^d\).
Hyperplane sections of \(\Gamma_\beta\) are exactly the elements of
\(\cU_\beta\) restricted to the unit circle.

Let
\begin{equation}\label{eq:delta-beta}
  \Delta_\beta(\theta)
  =\det(\gamma_\beta,\gamma_\beta',\ldots,\gamma_\beta^{(d)}).
\end{equation}
For an \(\mathbb R^N\)-valued curve \(C=C(r)\), write
\[
  \mathcal D_r C=\det(C,\partial_r C,\ldots,\partial_r^{N-1}C),
  \qquad J=\frac{N(N-1)}2.
\]
Triangularity of derivative jets gives
\begin{equation}\label{eq:jet-transformation-laws}
  \mathcal D_\theta(C\circ s)=(s')^J\bigl(\mathcal D_sC\bigr)\circ s,
  \qquad
  \mathcal D_r(hC)=h^N\mathcal D_rC,
  \qquad
  \mathcal D_r(AC)=\det(A)\mathcal D_rC,
\end{equation}
for a scalar function \(h\) and a constant matrix \(A\).  In derivative
order \(j\), the leading terms are respectively \((s')^jC^{(j)}\) and
\(hC^{(j)}\); all other terms involve earlier jet columns.  In particular,
\begin{equation}\label{eq:sphere-jet-factor}
  \mathcal D_\theta\!\left(\frac{\Gamma}{\|\Gamma\|}\right)
  =\|\Gamma\|^{-N}\mathcal D_\theta\Gamma.
\end{equation}

Fix \(\theta_0\) with \(\cos\theta_0\neq0\), and set
\[
  t=\tan\theta,\quad t_0=\tan\theta_0,\quad s=t-t_0,
  \quad c_i=Q_i(1,t_0)>0,
  \quad q_i(s)=c_i^{-1}Q_i(1,t_0+s).
\]
Let \(\Phi_\beta(s)\) be the column vector of the local paired functions
in \eqref{eq:local-pairs}, and put
\[
  B_u=\begin{pmatrix}1&0\\ t_0&1\end{pmatrix},
  \qquad
  P_u=\operatorname{diag}_{i=1}^n(c_i^{-\beta}B_u),
  \qquad
  h_u(\theta)=\cos\theta\,|\cos\theta|^{-2\beta}.
\]
Then \(\Gamma_\beta(\theta)=h_u(\theta)P_u
\Phi_\beta(s(\theta))\).  Since the Wronskian in
\cref{sec:local-wronskian} is the transpose of the corresponding
jet-column determinant, \eqref{eq:jet-transformation-laws} and
\eqref{eq:sphere-jet-factor} yield
\begin{equation}\label{eq:delta-wronskian-u-chart}
  \Delta_\beta(\theta_0)
  =\|\Gamma_\beta(\theta_0)\|^{-N}h_u(\theta_0)^N
    (\sec^2\theta_0)^J
    \left(\prod_{i=1}^n c_i^{-2\beta}\right)W_\beta(0).
\end{equation}
Translation has derivative one, quadratic normalization contributes the
positive product displayed above, and \(\det B_u=1\).

The complementary chart is equally explicit.  If
\(\sin\theta_0\neq0\), use \(t=\cot\theta\), \(s=t-t_0\),
\(c_i=Q_i(t_0,1)\),
\(q_i(s)=c_i^{-1}Q_i(t_0+s,1)\), and
\[
  B_v=\begin{pmatrix}t_0&1\\1&0\end{pmatrix},
  \qquad
  P_v=\operatorname{diag}_{i=1}^n(c_i^{-\beta}B_v),
  \qquad
  h_v(\theta)=\sin\theta\,|\sin\theta|^{-2\beta}.
\]
Then \(\Gamma_\beta(\theta)=h_v(\theta)P_v
\Phi_\beta(s(\theta))\).  Here
\(ds/d\theta=-\csc^2\theta\) and \(\det B_v=-1\), so
\begin{equation}\label{eq:delta-wronskian-v-chart}
  \Delta_\beta(\theta_0)
  =\|\Gamma_\beta(\theta_0)\|^{-N}h_v(\theta_0)^N
    (-\csc^2\theta_0)^J(-1)^n
    \left(\prod_{i=1}^n c_i^{-2\beta}\right)W_\beta(0).
\end{equation}
Because \(N=2n\), the possibly negative common factors \(h_u,h_v\) have
positive \(N\)-th powers.  Moreover \(J=n(2n-1)\equiv n\pmod2\), so in
\eqref{eq:delta-wronskian-v-chart} the reparameterization sign cancels the
sign from the \(n\) numerator-basis blocks.  Thus the multipliers in both
chart formulas are positive.  The two charts cover the circle, including
the projective point at infinity.  In each formula, \(W_\beta(0)\) denotes
the locally normalized Wronskian for the displayed chart.  Consequently
\cref{cor:wronskian-nonzero} gives
\(\Delta_\beta(\theta)\neq0\) for \(1<\beta\leq\beta_*\), while
\cref{prop:rational-normal} gives the same conclusion at \(\beta=1\).
Hence \(\Delta_\beta\) is continuous and nonzero on the connected set
\([1,\beta_*]\times(\R/2\pi\mathbb Z)\), and has one constant sign.

Alves uses the positive convention.  If the common sign is negative, fix
once and for all an orthogonal map \(O\in O(N)\) with \(\det O=-1\); if it
is positive, take \(O=I\).  Then
\begin{equation}\label{eq:positive-orientation}
  \det(O\gamma_\beta,O\gamma_\beta',\ldots,
       O\gamma_\beta^{(d)})>0
\end{equation}
throughout the exponent path.  This single fixed transformation does not
change any zero set or multiplicity, because for every real linear functional \(\ell_0\),
\(\ell_0(O\gamma)=(O^{\mathsf T}\ell_0)(\gamma)\); it also preserves global
projective convexity.  The orientation-reversing map \(O\) is not
Spin-lifted.  All Frenet and Spin constructions below occur only after
\eqref{eq:positive-orientation}.

For comparison with the fixed-\(\operatorname{Spin}\) formulation, fix \(c\) and use the positively
oriented interval \([c,c+\pi]\), reparameterized by
\(\theta=c+\pi t\), \(0\leq t\leq1\).  Let
\(F_{\beta,c}(\theta)\in SO(N)\) be the Frenet frame obtained from the
positive jet of \(O\gamma_\beta\) by Gram--Schmidt with positive diagonal.
Apply the fixed rotation \(F_{\beta,c}(c)^{-1}\) to normalize the initial
frame to \(I\).

Quadratic homogeneity gives
\begin{equation}\label{eq:antiperiodic}
  \Gamma_\beta(\theta+\pi)=-\Gamma_\beta(\theta).
\end{equation}
The norm is \(\pi\)-periodic, and the same antiperiodicity holds for every
derivative after the fixed orthogonal and initial-frame normalizations.
In particular,
\([O\gamma_\beta(\theta+\pi)]=[O\gamma_\beta(\theta)]\) in projective
space.
Thus the terminal jet is the negative of the initial jet.  Uniqueness of
positive-diagonal Gram--Schmidt and the evenness of \(N=2n\) give the
relative terminal frame
\begin{equation}\label{eq:terminal-minus-I}
  -I\in SO(N).
\end{equation}
Define the initial-normalized frame family
\[
  H(\beta,t)=F_{\beta,c}(c)^{-1}F_{\beta,c}(c+\pi t),
  \qquad (\beta,t)\in[1,\beta_*]\times[0,1].
\]
It is continuous, with \(H(\beta,0)=I\) and \(H(\beta,1)=-I\).
The homotopy-lifting property for the covering
\(\operatorname{Spin}(N)\to SO(N)\), applied with the fixed lift \(1\) on
the initial edge, gives a unique continuous lift \(\widetilde H\) on this
parameter rectangle.  Thus \(z_\beta=\widetilde H(\beta,1)\) depends
continuously on \(\beta\) and lies in the discrete two-point fiber over
\(-I\).  It is therefore constant:
\begin{equation}\label{eq:spin-constant}
  z_\beta=z
  \qquad(1\leq\beta\leq\beta_*).
\end{equation}
Thus the interval path lies in a single fixed
\(\cLS^{2n-1}(z)\).  This observation is not used in the projective argument below.

All required regularity is uniform on this finite exponent interval.  Each
\(Q_i\) has a positive minimum on the circle, and
\[
  Q_i(x(\theta))^{-\beta}
  =\exp[-\beta\log Q_i(x(\theta))]
\]
and all its \(\theta\)-derivatives are jointly continuous on the compact
rectangle.  Sphere normalization, Gram--Schmidt on the nondegenerate jet
locus, and initial-frame normalization preserve continuity.  Hence the
path is continuous in every finite \(C^r\) topology, in particular the
\(C^d\) topology used for smooth locally convex curves in
\cite{Alves2016}.

The sphere dimension is \(d=2n-1\geq3\).  Equation
\eqref{eq:spin-constant} records compatibility with the fixed-endpoint
formulation in \cref{thm:convex-component}; neither the fixed-\(\operatorname{Spin}\) theorem 
nor any choice of a cut or endpoint convention is used in the following projective argument.
Anti-periodicity
\eqref{eq:antiperiodic} gives a closed projective curve
\[
  \bar\gamma_\beta=[O\gamma_\beta]:
  \R/\pi\mathbb Z\longrightarrow\RP^d
\]
for every \(\beta\in[1,\beta_*]\).  A half-open interval
\([c,c+\pi)\) is a fundamental domain (equivalently, use
\([c,c+\pi]\) with its endpoints identified).  Under
\(\varphi=2\theta\), this is Anisov's ordinary parameter circle.  Reading the
same \(\theta\) modulo \(2\pi\) would instead traverse the projective curve
twice, which does not preserve convexity in Anisov's convention
\cite{Anisov1998}.  Since \(O\) is invertible, injectivity reduces to the
untransformed curve, whose first two-coordinate block recovers
\([x(\theta)]\).  
Thus the descended map is injective; the established nonvanishing 
of \(\Delta_\beta\) in \eqref{eq:delta-beta} makes it an immersion, hence an embedded closed curve.
The
determinant
\eqref{eq:delta-beta} says precisely that these curves have no flattenings,
and the regularity just verified is more than the \(C^{d+1}\) continuity
needed in \cref{cor:closed-projective-persistence}.  At \(\beta=1\),
\cref{prop:rational-normal} gives a globally convex rational normal curve;
the fixed orthogonal map and sphere normalization preserve its projective
hyperplane intersections.  Hence
\cref{cor:closed-projective-persistence} makes the whole exponent path
globally convex, in particular its endpoint at \(\beta_*\).

\begin{proof}[Proof of \cref{thm:paired-haar}]
Since \(\beta_*>1\) was arbitrary, the preceding conclusion holds for
every finite \(\beta>1\).  Let a nonzero hyperplane section be given.
The nonvanishing full-jet determinant implies that no nonzero hyperplane section vanishes identically, and analyticity on the
compact half-period circle \(\R/\pi\mathbb Z\) makes its zero set finite.
Global convexity gives at most \(d=2n-1\) projective zeros there, counted
with multiplicity.  Equivalently, one may choose a cut \(c\) away from the
zero set; then every zero has exactly one representative in
\((c,c+\pi)\), and neither endpoint is a zero.  
This choice of interval merely selects one representative of each projective zero; it is not used in applying Anisov's theorem.  
Together
with the linear independence supplied by
\cref{cor:wronskian-nonzero}, this proves the theorem.
\end{proof}

\section{Proof of the signed line theorem}
\label{sec:signed-proof}

We now return to the reduced data of \cref{sec:main-results}.  Put
\begin{equation}\label{eq:signed-beta}
  \beta=\alpha+1>1,
  \qquad
  E_0(t)=-\frac{F_{\rm red}'(t)}{2\alpha}
  =\sum_{j=1}^{m}c_j
  \frac{t-a_j}{\bigl((t-a_j)^2+b_j^2\bigr)^\beta}.
\end{equation}

\subsection{Exact cancellation}

\begin{lemma}[Exact cancellation classification]
\label[lemma]{lem:exact-cancellation}
The original restriction in \eqref{eq:intro-signed} satisfies
\begin{equation}\label{eq:exact-cancellation-classification}
  F_\alpha\equiv0\ \hbox{on }D_{\rm orig}
  \quad\Longleftrightarrow\quad
  \widehat\eta_C=0\ \hbox{for every }C\in\mathfrak C.
\end{equation}
Equivalently, its reduced extension is the zero function.  Thus distinct
reduced kernels are linearly independent, including at integer and
half-integer values of \(\alpha\).
\end{lemma}

\begin{proof}
The reverse implication is immediate after grouping.  Conversely, suppose
that \(F_\alpha\) vanishes on \(D_{\rm orig}\).  If a class with
\(b_C=0\) had nonzero class sum \(c\), then as \(t\to a_C\) its combined
term would be \(c|t-a_C|^{-2\alpha}\), while every other reduced class
would have a different pole or remain bounded there.  This is impossible.
Hence all zero-height class sums vanish.

Thus every active class has positive height.  The reduced sum \(F_{\rm red}\) is real analytic on \(\R\) and
vanishes on a nonempty open subinterval of \(D_{\rm orig}\).  Differentiating
there gives
\begin{equation}\label{eq:cancellation-derivative}
  -\frac{F_{\rm red}'(t)}{2\alpha}
  =\sum_j c_j\frac{t-a_j}
  {\bigl((t-a_j)^2+b_j^2\bigr)^{\alpha+1}}.
\end{equation}
The associated monic positive-definite quadratics are pairwise
nonproportional, and the right side is a paired section with
\(\beta=\alpha+1>1\).  This real-analytic paired section vanishes on an
open interval and hence vanishes identically on \(\R\).  By the dimension
assertion in \cref{thm:paired-haar}, its paired generators are linearly
independent.
It can vanish identically only when every remaining class sum is zero.
\end{proof}

\subsection{The smooth reduced count}

For \(b\ge0\), define
\begin{equation}\label{eq:physical-quadratics}
  Q_{a,b}(u,v)=(u-av)^2+b^2v^2.
\end{equation}
The normalization in \eqref{eq:physical-quadratics}, in which the coefficient of \(u^2\) is \(1\), permits no nontrivial scalar proportionality.

\begin{lemma}[Proportionality]\label[lemma]{lem:physical-proportionality}
If \(b,b'\geq0\)
and \(Q_{a,b}=cQ_{a',b'}\) for a real constant \(c\), then
\[
  c=1,
  \qquad a=a',
  \qquad b^2=(b')^2.
\]
In particular, for nonnegative heights the two kernels are identical.
\end{lemma}

\begin{proof}
Expanding gives
\[
  Q_{a,b}(u,v)=u^2-2auv+(a^2+b^2)v^2.
\]
Comparison of the coefficients of \(u^2\), \(uv\), and \(v^2\), in that
order, gives the three asserted equalities.
\end{proof}

\begin{proposition}[Smooth signed count]\label[proposition]{prop:smooth-signed-count}
If \(m\geq1\) and every effective \(b_j>0\), then
\eqref{eq:smooth-signed-bound} holds.
\end{proposition}

\begin{proof}
For each reduced class \(j=1,\ldots,m\), set
\begin{equation}\label{eq:smooth-section-data}
  Q_j(u,v)=(u-a_jv)^2+b_j^2v^2,
  \qquad L_j(u,v)=c_j(u-a_jv).
\end{equation}
The quadratics are positive definite and pairwise nonproportional by
\cref{lem:physical-proportionality}.  In the affine chart \(v=1\), the
paired section \(\sum_jL_j/Q_j^\beta\) is \(E_0\).  Since every \(c_j\)
is nonzero and the paired generators are independent, this section is
nonzero.  The projective multiplicity bound in \cref{thm:paired-haar} is
\(2m-1\).  Restricting that count to the finite affine chart proves
\eqref{eq:smooth-signed-bound}; any zero at infinity merely consumes part
of the same projective budget.
\end{proof}

\subsection{One global count in the singular case}

Now suppose at least one effective \(b_j=0\).  The domain is the punctured
set \(D_{\rm red}\) in \eqref{eq:reduced-domain}.  Write
\begin{equation}\label{eq:G-beta}
  G_\beta(a,b;t)=
  \frac{t-a}{\bigl((t-a)^2+b^2\bigr)^\beta},
\end{equation}
so that \(E_0=\sum_jc_jG_\beta(a_j,b_j;\cdot)\).

For an effective zero-height class with parameter pair \((a_j,0)\), put \(s=t-a_j\). The corresponding pole is at \(t=a_j\).
Reduction makes this
the unique singular kernel at \(a_j\), and all other terms are bounded
there.  Consequently
\begin{equation}\label{eq:pole-asymptotics}
  F_{\rm red}(a_j+s)=c_j|s|^{-2\alpha}+O(1),
  \qquad
  E_0(a_j+s)=c_j\sgn(s)|s|^{-2\alpha-1}+O(1).
\end{equation}
Thus each pole lies outside the domain, and a sufficiently small punctured neighborhood of 
every pole contains no zero of \(E_0\).  
On every component of
\(D_{\rm red}\), the field \(E_0\) is real analytic.  
Every component is adjacent to an effective pole, and the asymptotic in 
\eqref{eq:pole-asymptotics} precludes \(E_0\) from vanishing identically on that component.
All zeros in the domain are therefore isolated and have finite analytic
multiplicity.

Choose one effective pole \((a_*,0)\) with coefficient \(c_*\neq0\).  For
a small nonzero real \(\delta\), replace it by
\begin{equation}\label{eq:signed-pole-displacement}
  (a_{*,\delta},b_{*,\delta})
  =\bigl(a_*+\sgn(c_*)\delta,|\delta|\bigr),
\end{equation}
give every other zero-height class the height \(|\delta|\) without moving
its projection, and leave all positive heights unchanged.  Let
\(E_\delta\) be the resulting smooth field.  On every compact
\(K\Subset D_{\rm red}\) and for each fixed \(r\geq0\), Taylor expansion
in the projection parameter and in \(b^2\) gives
\begin{equation}\label{eq:singular-expansion}
  \left\|E_\delta-E_0-\delta h\right\|_{C^r(K)}
  \leq C_{K,r}\delta^2
\end{equation}
for all sufficiently small \(|\delta|\), where
\begin{equation}\label{eq:positive-first-variation}
  h(t)=(2\beta-1)|c_*|\,|t-a_*|^{-2\beta}>0.
\end{equation}
Indeed,
\begin{equation}\label{eq:first-variation-calculation}
  \left.
  \frac{\partial}{\partial a}G_\beta(a,0;t)
  \right|_{a=a_*}
  =(2\beta-1)|t-a_*|^{-2\beta},
\end{equation}
the product \(c_*\sgn(c_*)\) is \(|c_*|\), and every height change enters
through \(b^2=\delta^2\).  The expansion has the same positive first
variation for both signs of \(\delta\).

Suppose, for a contradiction, that \(E_0\) has at least \(2m\) distinct
zeros anywhere in \(D_{\rm red}\).  Select exactly \(2m\) of them across
all components, and call this finite set \(S\).  Choose pairwise disjoint
compact isolating intervals
\begin{equation}\label{eq:isolating-intervals}
  I_q=[q-\varepsilon_q,q+\varepsilon_q]\Subset D_{\rm red},
  \qquad q\in S,
\end{equation}
such that \(q\) is the only zero of \(E_0\) in \(I_q\).  Shrink them so
that an odd-multiplicity root has opposite endpoint signs; an
even-multiplicity root with positive first nonzero Taylor coefficient has
two positive endpoint values; and an even root with negative first
coefficient has two negative endpoint values.  Let \(O\), \(A\), and
\(B\) denote the respective numbers of selected odd, positive-oriented
even, and negative-oriented even roots.

Apply \eqref{eq:singular-expansion} to the one compact set
\begin{equation}\label{eq:global-compact-union}
  K=\bigcup_{q\in S}I_q.
\end{equation}
Because this is a finite union, there is one \(\delta_0>0\) such that,
simultaneously for every selected root and both perturbation signs, all
endpoint signs are preserved and
\begin{equation}\label{eq:even-root-center-sign}
  E_\delta(q)=\delta h(q)+O(\delta^2)
\end{equation}
has the sign of \(\delta\) whenever \(0<|\delta|<\delta_0\).  
Equivalently, there exists \(\delta_0>0\) such that, for every \(q\in S\) and 
every \(\delta\) with \(0<|\delta|<\delta_0\), the endpoint signs of \(I_q\) are preserved and \(E_\delta(q)\) has the sign of \(\delta\). 
The same \(\delta_0\) works for all \(q\in S\) and for both signs of \(\delta\).

Every odd root persists for either sign by its endpoint signs.  If
\(\delta<0\), each of the \(A\) positive-oriented even roots has negative
center value and positive endpoint values, and therefore gives one root on
each side of its center.  If \(\delta>0\), each of the \(B\)
negative-oriented even roots analogously gives two roots.  Choosing the
better sign yields at least
\begin{equation}\label{eq:majority-unfolding}
  O+2\max(A,B)\geq O+A+B=2m
\end{equation}
distinct roots of one smooth field, all in the selected disjoint intervals.

Only finitely many nonzero values of \(\delta\) can create a collision 
among the perturbed parameter pairs. Choose the sign selected above and 
then choose a common \(\delta\) with \(0<|\delta|<\delta_0\) outside this finite set. 
The resulting field has \(m\) pairwise nonproportional 
positive-definite quadratic denominators, so 
\cref{prop:smooth-signed-count} gives at most \(2m-1\) real zeros counted with multiplicity.
This contradicts
\eqref{eq:majority-unfolding} and proves the global estimate
\eqref{eq:singular-signed-bound}.  
Notice that the proof selected roots across all 
components, used one compact union and one common perturbation, and did not 
obtain the global estimate by summing separate bounds over the connected components of \(D_{\rm red}\).

\begin{remark}[Refined singular bound]
\label{rem:count-singular}
In the singular case, let \(r=\#\{j:b_j=0\}\) be the number of effective poles. The preceding argument gives the sharper estimate
\[
\sum_{I\in\pi_0(D_{\rm red})}
\#\{t\in I:F'_{\rm red}(t)=0\}\le 2m-r-1.
\]
Indeed, choose any finite set \(S\) of zeros of \(E_0\) in \(D_{\rm red}\), together with disjoint isolating intervals \(I_q\) as above. 
Around each effective pole choose a sufficiently small closed interval, with endpoints in \(D_{\rm red}\), 
such that these \(r\) intervals are mutually disjoint and avoid all \(I_q\). 
By \eqref{eq:pole-asymptotics}, \(E_0\) has opposite signs at the two endpoints of each pole interval. 
Convergence at these endpoints preserves their signs for both signs of sufficiently small \(\delta\), 
so \(E_\delta\) has at least one zero in each pole interval. 
The same choice between the two perturbation signs as in \eqref{eq:majority-unfolding} gives at least \(|S|\) additional zeros in the intervals \(I_q\). 
Choosing one common sufficiently small \(\delta\) and avoiding parameter collisions, Proposition~\ref{prop:smooth-signed-count} therefore gives
\[
|S|+r\le 2m-1.
\]
Since \(S\) was arbitrary, the claimed bound follows. In particular, if all effective sources lie on the line, the bound is \(m-1\).
\end{remark}

\subsection{Multipole tails and the point at infinity}
\label{sec:multipole-tails}

For the finite reduced data, choose \(R>0\) so large that every binomial
series below converges whenever \(|t|>R\).  For each reduced parameter pair
define coefficients \(P_r(a,b)\) by
\begin{equation}\label{eq:multipole-expansion}
  \bigl((t-a)^2+b^2\bigr)^{-\alpha}
  =|t|^{-2\alpha}\sum_{r\geq0}P_r(a,b)t^{-r}.
\end{equation}
The first four are
\begin{align}
  P_0(a,b)&=1,\label{eq:multipole-P0}\\
  P_1(a,b)&=2\alpha a,\label{eq:multipole-P1}\\
  P_2(a,b)&=\alpha\bigl((2\alpha+1)a^2-b^2\bigr),
    \label{eq:multipole-P2}\\
  P_3(a,b)&=2\alpha(\alpha+1)a
  \left(\frac{2\alpha+1}{3}a^2-b^2\right).
    \label{eq:multipole-P3}
\end{align}
Set
\begin{equation}\label{eq:multipole-moments}
  M_r=\sum_{j=1}^{m}c_jP_r(a_j,b_j).
\end{equation}
Then
\begin{equation}\label{eq:reduced-multipole-expansion}
  F_{\rm red}(t)=|t|^{-2\alpha}\sum_{r\geq0}M_rt^{-r}
  \qquad(|t|>R).
\end{equation}
If \(r_0\) is the first index with \(M_{r_0}\neq0\), termwise
differentiation gives
\begin{equation}\label{eq:multipole-derivative-tail}
  F_{\rm red}'(t)
  \sim-(2\alpha+r_0)M_{r_0}|t|^{-2\alpha}t^{-r_0-1}
  \qquad(t\to\pm\infty).
\end{equation}
Thus the derivative is nonzero on both sufficiently remote tails, even when
the total charge and several lower moments vanish.  If every \(M_r\)
vanished, the convergent expansion would make \(F_{\rm red}=0\) on
\((R,\infty)\).  
If an effective zero-height class existed, choose 
one with largest center and denote its parameter pair and coefficient by \((a_{\max},0)\) and \(c_{\max}\), respectively.
By maximality no effective pole lies in
\((a_{\max},\infty)\), so every reduced summand, and hence \(F_{\rm red}\), is
real analytic there.  The identity theorem would extend the tail identity to
that whole component.  But, writing \(t=a_{\max}+s\), reduction makes
\(c_{\max}\neq0\), and no other reduced class has the pole
\((a_{\max},0)\).  Thus every other term is bounded and
\[
  F_{\rm red}(a_{\max}+s)
  =c_{\max}s^{-2\alpha}+O(1)
  \qquad(s\downarrow0),
\]
with \(c_{\max}\neq0\), a contradiction for every \(\alpha>0\), regardless
of the sign of \(c_{\max}\).  
Thus, under the supposition that every \(M_r\) vanishes, no effective 
zero-height class can exist. All effective heights are then positive, so 
\(F_{\rm red}\) is real analytic on \(\mathbb R\); since it vanishes 
on \((R,\infty)\), the identity theorem gives \(F_{\rm red}\equiv0\) on \(\mathbb R\). 
\cref{lem:exact-cancellation} then forces \(m=0\). Consequently, 
when \(m\geq1\), not all moments \(M_r\) vanish, and \(r_0\) exists.

This tail calculation makes explicit that no affine roots accumulate at
infinity.  
In the smooth case this requires no additional zero-counting argument: any zero 
at infinity arising from moment cancellation is already counted, with its multiplicity, in the projective bound \(2m-1\).

\begin{proof}[Proof of \cref{thm:signed-line}]
The equivalences in \eqref{eq:exact-cancellation-equivalences} follow from
\cref{lem:exact-cancellation} together with the definition of \(m\).
If \(m=0\), the derivative vanishes at every point of \(D_{\rm orig}\).
For \(m\geq1\), \cref{prop:smooth-signed-count} proves the smooth
multiplicity assertion, and the globally coordinated perturbation above
proves the singular distinct-point assertion.  The inclusion in
\eqref{eq:domain-inclusion} transfers either estimate to the original
domain.  Finally \(m\leq|\mathfrak C|\leq\ell\), giving
\eqref{eq:presentation-bound}.
\end{proof}

\section{Point-charge consequences}
\label{sec:physical-consequences}
\label{sec:physical}

Let \(A_1,\ldots,A_n\in\R^3\), let \(c_i>0\), and define the positive
inverse-power potential
\begin{equation}\label{eq:intro-potential}
  V_p(x)=\sum_{i=1}^n\frac{c_i}{|x-A_i|^p},
  \qquad p>0.
\end{equation}
Parameterize a prescribed line \(\Lambda\) at unit speed by
\[
  x(t)=x_0+te,
  \qquad |e|=1.
\]
For each source, set
\begin{equation}\label{eq:line-geometry}
  a_i=\langle A_i-x_0,e\rangle,
  \qquad
  b_i=\dist(A_i,\Lambda).
\end{equation}
Then
\begin{equation}\label{eq:line-distance}
  |x(t)-A_i|^2=(t-a_i)^2+b_i^2,
\end{equation}
and direct differentiation gives
\begin{equation}\label{eq:physical-derivative}
  \frac{d}{dt}V_p(x(t))=-p E_p(t),
  \qquad
  E_p(t)=\sum_{i=1}^n c_i
  \frac{t-a_i}{\bigl((t-a_i)^2+b_i^2\bigr)^\beta},
  \qquad
  \boxed{\beta=1+\frac p2}.
\end{equation}
Thus \(x(t)\) is a critical point of \(V_p|_\Lambda\) in its regular domain if and only if \(E_p(t)=0\).

\begin{theorem}[Positive inverse-power Maxwell-line theorem]
\label{thm:physical}
Let \(p>0\), let \(A_1,\ldots,A_n\in\R^3\), let \(c_i>0\), and let
\(\Lambda\) be a straight line.  On
\begin{equation}\label{eq:positive-regular-domain}
  D=\Lambda\setminus\{A_i:A_i\in\Lambda\},
\end{equation}
the restriction of \eqref{eq:intro-potential} has at most \(2n-1\)
distinct critical points.  If no source lies on \(\Lambda\), the same
bound holds with analytic multiplicity.
\end{theorem}

\begin{proof}
Use \eqref{eq:line-geometry} and put \(\alpha=p/2>0\).  
Positive
coefficients prevent class cancellation, though duplicate restricted
kernels may combine, so the effective count satisfies \(m\leq n\).
In the smooth case, write \(L_i(u,v)=c_i(u-a_iv)\).  The corresponding
projective section also satisfies
\begin{equation}\label{eq:physical-infinity}
  Q_{a_i,b_i}(1,0)=1,
  \qquad
  \sum_i L_i(1,0)=\sum_i c_i>0,
\end{equation}
so it has no zero at infinity.
The applicable alternative of \cref{thm:signed-line} gives at most
\(2m-1\leq2n-1\) critical points.  If all sources are off the line, every
effective height is positive and the smooth multiplicity statement applies.
\end{proof}

\begin{corollary}[EFO Conjecture 3]\label{cor:efo}
For positive charges written in the EFO form \(c_i=\zeta_i^p>0\), the
substitution \(p=2\alpha\) identifies their line derivative with the
positive specialization of \cref{thm:signed-line}.  
Consequently, Conjecture 3 of Edelsbrunner--Fillmore--Oliveira holds throughout its stated range 
\(p\geq1\); moreover, the same conclusion holds for every \(p>0\).
\end{corollary}

For the planar Gabrielov--Novikov--Shapiro (GNS) restriction to the \(x\)-axis, the parameters 
are \(a_i=x_i\) and \(b_i^2=y_i^2\). Thus terms \(i\) and \(j\) define the same restricted 
kernel precisely when \((x_i,y_i^2)=(x_j,y_j^2).\)

\begin{corollary}[Natural reduced GNS formulation]\label{cor:gns}
In the signed setting of GNS Conjecture~1.9, combine equal restricted
kernels and delete zero class sums.  Exact cancellation occurs precisely
when every class sum is zero.  Otherwise, if \(m\) effective classes
remain, there are at most \(2m-1\) regular critical points with the smooth
and singular counting conventions of \cref{thm:signed-line}.  
Thus, outside the exact-cancellation case, the reduced GNS bound 
holds throughout the GNS range \(\alpha\geq\tfrac12\), and in fact for every \(\alpha>0\).
\end{corollary}

The word \emph{line} is essential in both corollaries.  Neither statement
asserts the unrestricted multidimensional Maxwell conjecture.

\section{Examples and positive optimality}
\label{sec:examples}
\label{sec:sharpness}

The following examples display cancellation, mixed signs, a zero total
coefficient, and a source on the line before we establish the positive
extremal construction.

\begin{example}[Reflected exact cancellation]\label{ex:reflection}
Restrict to the \(x\)-axis and take sites \((0,1)\) and \((0,-1)\), with
coefficients \(1\) and \(-1\).  Both sites induce the pair
\((a,b^2)=(0,1)\), so their class sum is zero and
\[
  F_\alpha(t)=\frac1{(t^2+1)^\alpha}
              -\frac1{(t^2+1)^\alpha}=0
\]
for every \(\alpha>0\).  
Here \(m=0\), and every ordinary point is
critical.  This is the literal exception that necessitates the reduced GNS
formulation.
\end{example}

\begin{example}[A nonzero mixed-sign restriction]\label{ex:mixed-sign}
At \(\alpha=1\), consider
\begin{equation}\label{eq:mixed-sign-example}
  F(t)=\frac{2}{t^2+1}-\frac{1}{(t-2)^2+1}.
\end{equation}
Its two kernels are distinct, so \(m=2\), and
\begin{equation}\label{eq:mixed-sign-derivative}
  F'(t)=
  \frac{-2P(t)}{(t^2+1)^2(t^2-4t+5)^2},
  \qquad
  P(t)=t^5-14t^4+50t^3-76t^2+49t+2.
\end{equation}
The exact values are
\[
  \bigl(P(-1),P(0),P(2),P(3),P(9),P(10)\bigr)
  =(-188,2,4,-76,-2068,2892).
\]
These sign changes give a zero of \(F'\) in each of
\((-1,0)\), \((2,3)\), and \((9,10)\).  The smooth paired bound is three
counted with multiplicity.  Hence these are exactly the three real critical
points, one in each displayed bracket, and all are simple.
\end{example}

\begin{example}[Zero total coefficient and an infinity zero]
\label{ex:zero-total}
Still at \(\alpha=1\), take
\begin{equation}\label{eq:zero-total-example}
  F(t)=\frac{1}{t^2+1}-\frac{1}{(t-2)^2+1}.
\end{equation}
The total coefficient is zero, but the two kernels are distinct and
\(m=2\).  Direct calculation gives
\begin{equation}\label{eq:zero-total-derivative}
  F'(t)=
  \frac{4\bigl(3(t-1)^4-4\bigr)}
  {(t^2+1)^2(t^2-4t+5)^2}.
\end{equation}
Thus the two finite real critical points are
\begin{equation}\label{eq:zero-total-roots}
  t=1\pm\left(\frac43\right)^{1/4},
\end{equation}
and both are simple.  
For the paired section \(E=-F'/2\), the representative in the complementary 
projective chart is \(E_\infty(s)=s^{-3}E(1/s)\), which 
extends analytically to \(s=0\) and satisfies \(E_\infty(s)=-6s+O(s^2)\).
It therefore has a simple infinity zero:
the zeroth numerator sum vanishes, while the next coefficient does not.
Thus the two finite roots together with the infinity zero have total
projective multiplicity three, exactly the \(2m-1\) budget.
\end{example}

\begin{example}[One effective source on the line]\label{ex:on-line-one}
For \(F_\alpha(t)=|t|^{-2\alpha}\), the reduced domain is
\(\R\setminus\{0\}\), \(m=1\), and
\begin{equation}\label{eq:on-line-one-derivative}
  F_\alpha'(t)=-2\alpha\sgn(t)|t|^{-2\alpha-1}.
\end{equation}
There is no critical point in either component.  The pole \(t=0\) is
outside the domain and is not critical; 
Remark \ref{rem:count-singular} gives the sharp bound \(2m-r-1=0\).
\end{example}

The next result concerns positive configurations only.  It proves that 
the \(2n-1\) bound in terms of the original number \(n\) of charges cannot be lowered uniformly; it
does not assert a separate signed optimality theorem.

\begin{theorem}[Positive universal sharpness]\label{thm:sharpness}
For every \(p>0\) and \(n\geq1\), take the line to be the \(x\)-axis and
place unit positive charges at
\begin{equation}\label{eq:sharp-sites}
  A_i=\left(i,\frac1{4n},0\right),
  \qquad i=1,\ldots,n.
\end{equation}
The restriction of \(V_p\) to the line has exactly \(2n-1\) distinct
critical points, all simple.  Consequently
\begin{equation}\label{eq:positive-extremal-count}
  \max\#\Eq(V_p|_\Lambda)=2n-1
  \qquad(p>0),
\end{equation}
where the maximum ranges over positive \(n\)-charge configurations and
straight lines.
Here, \(\operatorname{Eq}(V_p|_\Lambda)\) denotes the set of critical points of \(V_p|_\Lambda\) in its regular domain.
\end{theorem}

\begin{proof}
Put
\begin{equation}\label{eq:sharp-family}
  a_i=i,
  \qquad b=\frac1{4n},
  \qquad c_i=1,
  \qquad i=1,\ldots,n,
\end{equation}
and define
\begin{equation}\label{eq:sharp-field}
  E_{p,n}(t)=\sum_{j=1}^n
  \frac{t-j}{\bigl((t-j)^2+b^2\bigr)^{1+p/2}}.
\end{equation}
Along the \(x\)-axis, \(V_p'=-pE_{p,n}\).

Assume first that \(n\geq2\).  At \(t=i\pm b\), the self term has
magnitude
\begin{equation}\label{eq:self-magnitude}
  S_{p,n}=2^{-1-p/2}(4n)^{p+1}.
\end{equation}
For \(j\neq i\),
\[
  |t-j|\geq1-b\geq\frac34,
\]
and, with \(\beta=1+p/2\),
\begin{equation}\label{eq:tail-single}
  \frac{|t-j|}{\bigl((t-j)^2+b^2\bigr)^\beta}
  <|t-j|^{1-2\beta}
  \leq\left(\frac43\right)^{p+1}.
\end{equation}
The self term dominates the sum of the remaining terms because
\begin{equation}\label{eq:sharp-ratio}
  \frac{S_{p,n}}{(n-1)(4/3)^{p+1}}
  =\frac{3n}{2(n-1)}
    \left(\frac{3n}{\sqrt2}\right)^p
  >\frac32>1.
\end{equation}
Therefore
\begin{equation}\label{eq:sharp-signs}
  E_{p,n}(i-b)<0<E_{p,n}(i+b)
  \qquad(i=1,\ldots,n).
\end{equation}
Since \(2b=1/(2n)<1\), the \(2n\) test points are ordered as
\[
  1-b<1+b<2-b<2+b<\cdots<n-b<n+b.
\]
Their signs alternate.  The intermediate value theorem supplies a zero in
each of the \(2n-1\) disjoint intervals between consecutive test points.

The smooth paired Haar theorem bounds the sum of the multiplicities of all
zeros by \(2n-1\).  We have already found that many distinct zeros, each of
multiplicity at least one.  Hence there are exactly \(2n-1\) zeros, and
each is simple.  When \(n=1\), the unique zero is \(t=1\), and
\[
  E_{p,1}'(1)=b^{-2-p}>0,
\]
so it is simple as well.
\end{proof}

The results concern \(p>0\). At \(p=0\), \(V_0\) is constant on its regular domain, and no assertion is made about the limit \(p\to\infty\).

\subsection*{Disclosure of AI assistance}
The main idea of this work emerged through exploratory conversations with ChatGPT using GPT-5.6 Sol Pro. 
During the subsequent development, the author used the system as an aid in developing the arguments and 
improving their presentation. 
GPT-6 Astra was also used to review the manuscript.
The author independently reconstructed, checked, and 
revised the mathematical details and assumes full responsibility for the mathematical correctness and final manuscript.

\renewcommand{\bibliofont}{\scriptsize}\bibliographystyle{amsplain}
\bibliography{references}

\end{document}